\documentclass{l4dc2024}

\usepackage{times}
\usepackage{float}
\usepackage{hyperref}
\usepackage{algorithmic}
\usepackage{subcaption}
\usepackage{graphicx}
\usepackage{adjustbox}
\usepackage{multirow}
\usepackage{pbox}
\usepackage{algorithm}
\usepackage{pgfplots}
\newcommand{\bx}{{\boldsymbol{x}}}
\newcommand{\bu}{{\boldsymbol{u}}}
\newcommand{\by}{{\boldsymbol{y}}}
\newcommand{\bz}{{\boldsymbol{z}}}
\newcommand{\bd}{{\boldsymbol{d}}}
\newcommand{\bw}{{\boldsymbol{w}}}
\newcommand{\bv}{{\boldsymbol{v}}}

\usepgfplotslibrary{groupplots,dateplot}
\usetikzlibrary{patterns,shapes.arrows}
\pgfplotsset{compat=newest}
\usepackage{hyperref}

\usepackage{adjustbox}
 \usepackage{multirow}
\usepackage{pbox}
\usepackage{tikz}
\usetikzlibrary{shapes, arrows.meta, positioning}

\usepgfplotslibrary{groupplots}
\usepgfplotslibrary{fillbetween}
\usetikzlibrary{arrows,decorations.pathmorphing,positioning,fit,trees,shapes,shadows,automata,calc} 
\usetikzlibrary{patterns,arrows,arrows.meta,calc,shapes,shadows,decorations.pathmorphing,decorations.pathreplacing,automata,shapes.multipart,positioning,shapes.geometric,fit,circuits,trees,shapes.gates.logic.US,fit, matrix,arrows.meta, quotes}
\usetikzlibrary{backgrounds,scopes}
\usepackage{ulem}

\usepackage{tabu}

\title[Neural Integral Control Barrier Functions]{Neural Differentiable Integral Control Barrier Functions for Unknown Nonlinear Systems with Input Constraints}
\usepackage{times}

\author{%
 \Name{Vrushabh Zinage} \Email{vrushabh.zinage@utexas.edu}\\
 \addr Aerospace Engineering, University of Texas at Austin, Texas, USA
 \AND
\Name{Rohan Chandra} \Email{rchandra@utexas.edu}\\
 \addr Computer Science, University of Texas at Austin, Texas, USA
 \AND
 \Name{Efstathios Bakolas} \Email{bakolas@austin.utexas.edu}\\
 \addr Aerospace Engineering, University of Texas at Austin, Texas, USA
}

\begin{document}

\newcommand{\fknown}{\mathrm{F}}
\newcommand{\funk}{\mathrm{g}}
\newcommand{\constr}{\Psi}

\newcommand{\state}{x}
\newcommand{\stateSpace}{\mathcal{X}}
\newcommand{\dynFun}{f}
\newcommand{\controlInput}{u}
\newcommand{\controlSpace}{\mathcal{U}}
\newcommand{\trajectory}{\tau}
\newcommand{\stateDim}{n}
\newcommand{\controlDim}{m}

\newcommand{\unknownTerm}{g}
\newcommand{\vectorField}{F}
\newcommand{\numUnknownTerms}{d}

\newcommand{\nnParamsAll}{\Theta}
\newcommand{\nnParams}{\theta}
\newcommand{\numParams}{k}

\newcommand{\predNextState}{\Gamma}
\newcommand{\loss}{\mathcal{J}}
\newcommand{\dataset}{\mathcal{D}}
\newcommand{\datasetDim}{|\dataset|}
\newcommand{\numTrainingTrajectories}{i}
\newcommand{\trajectoryTimeHorizon}{T}
\newcommand{\rollout}{n_r}
\newcommand{\odesolve}{\mathrm{ODESolve}}

\newcommand{\inequalConstr}{\Psi}
\newcommand{\equalConstr}{\Phi}
\newcommand{\constrDomain}{\mathcal{C}}
\newcommand{\numInequalConstr}{l}
\newcommand{\numEqualConstr}{v}
\newcommand{\collocationPoints}{\Omega}
\newcommand{\numCollocationPoints}{|\collocationPoints|}
\newcommand{\augmentedLagrangian}{\mathcal{L}}
\newcommand{\constrPenalty}{\mu}
\newcommand{\lagrangeVar}{\lambda}
\newcommand{\lagrangeVarAll}{\Lambda}
\newcommand{\totalNumEqualConstr}{N_{\equalConstr}}
\newcommand{\totalNumInequalConstr}{N_{\inequalConstr}}

\newcommand{\collectionUnknownTerms}{G_{\nnParamsAll}}

\newcommand{\constrTol}{\epsilon}

\newcommand{\massMat}{M}
\newcommand{\corForce}{C}
\newcommand{\actuationForce}{\tau}
\newcommand{\contactForce}{\mathcal{F}}
\newcommand{\jacobian}{J}

\newcommand{\rohan}[1]{\textbf{\textcolor{red}{$\leftarrow$ #1}}}
\newcommand{\rc}[1]{\textbf{\textcolor{blue}{#1}}}
\maketitle

\begin{abstract}%
In this paper, we propose a deep learning based control synthesis framework for fast and online computation of controllers that guarantees the safety of general nonlinear control systems with unknown dynamics in the presence of input constraints. Towards this goal, we propose a framework for simultaneously learning the unknown system dynamics, which can change with time due to external disturbances, and an integral control law for trajectory tracking based on imitation learning. Simultaneously, we learn corresponding safety certificates, which we refer to as Neural Integral Control Barrier Functions (Neural ICBF's), that automatically encode both the state and input constraints into a single scalar-valued function and enable the design of controllers that can guarantee that the state of the unknown system will never leave a safe subset of the state space. Finally, we provide numerical simulations that validate our proposed approach and compare it with classical as well as recent learning based methods from the relevant literature.
\end{abstract}


\section{Introduction}


Modern self-driving cars \cite{hussain2018autonomous_cars}, factory robots, and multi-robot setups \cite{beltrame2018engineering_multi_robot_systems_1} used in unpredictable areas frequently need to adhere to vital safety measures while pursuing a specific task \cite{wabersich2023data_safety_systems}. Particularly, modern autonomous driving technology capitalizes on the engineered and consistent interactions between vehicles and roads, making them predictable for control tasks. However, in the case of off-road autonomy, natural terrains do not offer the same structured predictability, introducing challenges like vegetation, uneven surfaces, limited visibility, obstacles, and variable terrain properties. In such conditions, navigating at high speeds causes conventional perception, planning, and control methods to become ineffective. 
Consequently, there is a need to develop fast online learning paradigms that can easily adapt to changing dynamics, are computationally inexpensive and are able to generate safe control inputs.


One of the most popular methods for designing controllers that can guarantee safety is through the use of Control Barrier Functions (CBF's). Real-world CBF safety applications assume control-affine nonlinear systems with quadratic objective functions (QP's) with respect to the control inputs. 
However, there exist two significant challenges associated with this CBF-centric QP approach. First, the system's dynamics must be control affine, otherwise, the optimization process would involve a series of nonlinear programs (NLP's) that are both computationally expensive to solve (with no guarantees of convergence to global or local minimizers of the problem) and difficult to analyze \cite{son2019safety_cbf_nlp_1,seo2022safety_cbf_nlp_2,agrawal2017discrete_cbf_nlp_3}. Second, the sequence of QP's may not always be recursively feasible, especially in the presence of input constraints. Although the problem of recursive feasibility \cite{zeng2021enhancing_recursive_1,xiao2023feasibility_recursive_2,xiao2023learning_recursive_3,liu2023feasibility_recursive_4} and handling input constraints \cite{agrawal2021safe_input_constrained_1,breeden2021high_input_constraints_2,breeden2023robust_input_constrained_3,black2023consolidated_input_constrained_4} have been addressed in the literature, these works rely on three strong assumptions, namely that a valid CBF is known apriori, the nonlinear system has a control-affine structure, and a stabilizing feedback controller is known a-priori.

A common assumption is that the robotic systems that leverage CBFs to guarantee safety perfectly know which areas in their environment have been deemed unsafe and which are safe. Utilizing this knowledge can result in concrete safety assurances when converted to CBFs. However, in real-world scenarios, this presumption might not always be valid, hindering the broader use of barrier functions. For instance, imagine a robot functioning in a setting where it is not aware of obstacle boundaries. If we view these limits as level sets of continuously differentiable functions, deriving closed-form expression for the barrier functions corresponding to these barriers can be a very complex task. In the absence of these functions, the safety benefits of CBFs cannot be exploited. To address this issue, \cite{srinivasan2020synthesis_svm_cbf} proposes the use of state vector machines (SVMs) that learn the safe regions using the sensor measurements. However, these learned CBFs are restricted to only control affine systems and also cannot handle input constraints in general.
\subsection{Main contributions}

The main contributions of the paper are two fold. First, we propose Neural Integral Control Barrier Functions (NICBFs) for joint encoding state/safety and input constraints for general unknown nonlinear systems. Second, we propose a real-time control synthesis framework that is based on the combination of imitation learning-based control with the learned NICBF to synthesize safe control inputs for the nonlinear systems with unknown dynamics (including systems that are not control-affine). Furthermore, we provide an upper bound for the absolute error between the learned control inputs and the actual inputs that guarantee safety. Finally, we compare our proposed approach with well-established baselines such as Nonlinear MPC, CADRL \cite{cadrl}, ORCA-MAPF \cite{orcamapf} to show its efficacy in terms of computational time and cost required for the generated trajectories. 
\section{Related work\label{sec:related_work}}
CBF-based traditional approaches fall into two primary categories: model-based techniques and data-driven methods \cite{qin2021learning_neural_cbf_1,ahmadi2019safe_neural_cbf_2,yu2023learning_neural_cbf_3}. Model-based strategies \cite{ames2019control_ames_cbf_start,zinage2023disturbance} require a-priori understanding of the environment and set the CBF parameters in advance. These CBF constraints, being affine with respect to the control variable, pave the way for a quadratic program (QP) controller that ensures safe navigation. On the other hand, data-driven methods leverage neural networks \cite{dawson2022safe_neural_cbf_4,taylor2020learning_linear_cbf_3,zinage2023neural_koopman_cbf} to estimate/learn CBFs directly.
Approaches that utilize CBF's to guarantee safety and guide the exploration in episodic supervised learning of uncertain linear dynamics \cite{cheng2019end_linear_cbf_1,wang2018safe_linear_cbf_2,taylor2020learning_linear_cbf_3,taylor2020control_linear_cbf_4}. Most of these methods operate on the presumption that an appropriate CBF is already given. They can be considered as supplementary to our findings. In \cite{yaghoubi2020training_dnn_cbf}, a strategy rooted in imitation learning is employed to train a deep neural network (DNN) in order to emulate a CBF-centric controller.
\cite{xiao2021barriernet} uses BarrierNet coupled with imitation learning to compute safe control synthesis for nonlinear systems. However, their approach is restricted to known system dynamics and assumes that a valid CBF is known apriori. There have also been several other works on neural CBFs. However, their applications are complementary to this work. For instance, they have been employed to learn the unknown safety criteria from expert trajectories \cite{robey2020learning_other_neural_cbf_1,boffi2021learning_other_neural_cbf_2}, or to simultaneously learn a safe policy alongside its safety certificate in a reinforcement learning (RL) context \cite{qin2021learning_other_neural_cbf_3,meng2021reactive_other_neural_cbf_4,wang2021learning_other_neural_cbf_5}. 


On the other hand, Imitation Learning (IL) corresponds to a machine learning paradigm where an autonomous agent strives to learn a behavior by emulating an expert's demonstrations. IL-based approaches are much faster than traditional optimization-based approaches. Typically, the demonstrations comprise state-input pairs, generated by an expert policy during real-world execution. 
However, a shared drawback among the IL algorithms \cite{Kostrikov2019ImitationLearningViaOffPolicyDistributionMatching, Dadashi2020PrimalWassersteinImitationLearning,zheng2022imitation,bco,ifovideo} is their inability to encode state/safety and input constraints
during the demonstration phase. 
\begin{figure}
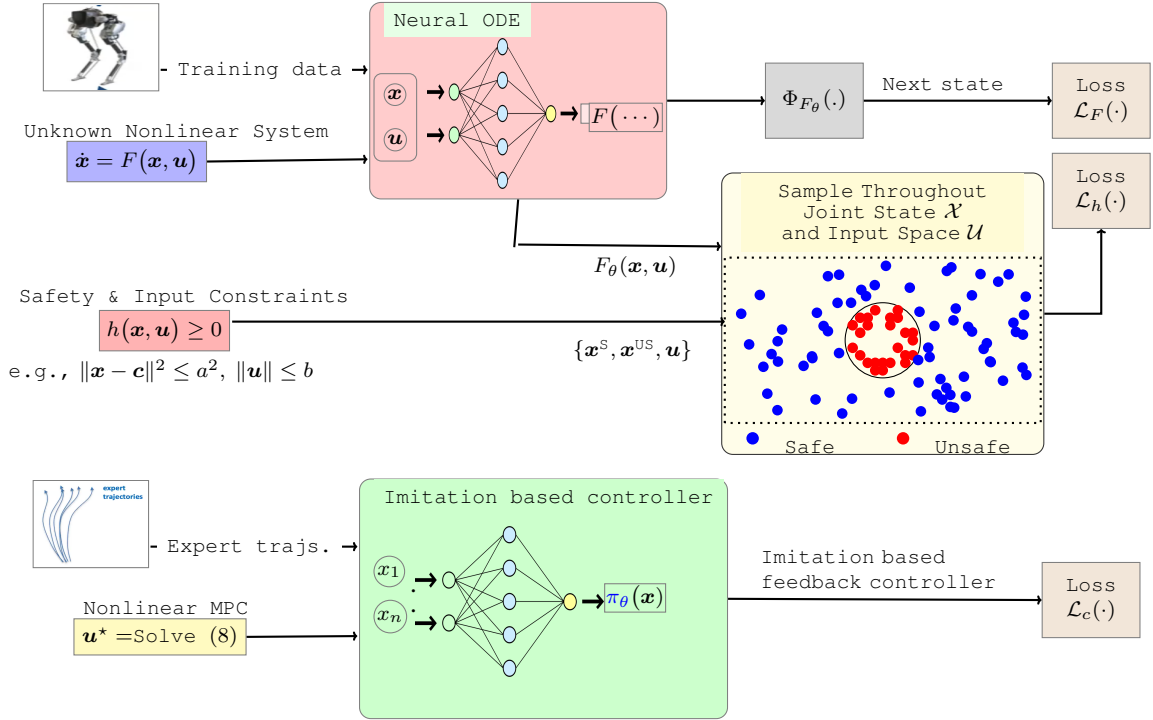

\input{figs/icbf_main}
    
    \vspace{-0.2cm} 
    
\input{figs/imitation_block}
\vspace{-1.5cm}
\caption{  \small  Proposed framework for joint learning of Integral Control Barrier Functions (ICBFs) $h_\theta(\bx,\bu)$, unknown nonlinear system $F_\theta(\bx,\bu)$ using Neural ODE's and imitation-based controller $\pi_\theta(\bx)$ using expert trajectories from Nonlinear MPC. Using these learned models, we synthesize safe controllers \eqref{eqn:u_safe_neural} that guarantee safety for general nonlinear systems in the presence of input constraints (i.e. $\bu\in\mathcal{U}$).}
\end{figure}

\section{Preliminaries and Problem Statement\label{sec:problem_statement}}


Consider the following general nonlinear controlled system given by
\begin{align}
    \dot{\bx}=F(\bx,\bu),\quad\quad \bx(0)=\bx^0
    \label{eqn:nonlinear_system}
\end{align}
where $\bx\in\mathcal{X}\subset\mathbb{R}^n$, $\bu\in\mathcal{U}\in\mathbb{R}^m$, $F:\mathcal{X}\times\mathcal{U}\rightarrow\mathbb{R}^n$ is a continuously differentiable function, $\mathcal{X}$ and $\mathcal{U}$ are compact sets. Now, consider a scalar-valued function $b(\bx):\mathcal{X}\rightarrow\mathbb{R}$ such that the following holds true
\begin{align}
    b(\bx)>0,\;\forall\; \bx\in\text{int}(\mathcal{S}_\bx),\quad b(\bx)=0,\;\forall \;\bx\in\delta\mathcal{S}_\bx,\quad b(\bx)<0,\;\forall \;\bx\in\mathcal{X}\setminus\mathcal{S}_\bx
    \label{eqn:cbf_safety_sets}
\end{align}
where $\text{int}(\mathcal{S}_\bx)$ and $\delta\mathcal{S}_\bx$ denotes the interior and boundary of the safe set $\mathcal{S}_\bx\subset\mathcal{X}$ respectively. The Control Barrier Function (CBF) for \eqref{eqn:nonlinear_system} is defined as follows
\begin{definition}
    \normalfont A function $b:\mathcal{X}\rightarrow\mathbb{R}$ which satisfies the conditions in \eqref{eqn:cbf_safety_sets} is said to be a Control Barrier Function (CBF) if there exists a $\mathcal{K}_\infty$\footnote{A continuous function $\gamma$ is said to be class-$\mathcal{K}_\infty$ function if it is continuously increasing, $\gamma(0)=0$ and $\underset{x\rightarrow\infty}{\lim}\;\gamma(x)=0$.} function $\gamma$ such that the following holds true
    \begin{align}
        \underset{\bu\in\mathcal{U}}{\inf}\;\dot{b}(\bx):=\underset{\bu\in\mathcal{U}}{\inf}\;\frac{\partial b}{\partial \bx}F(\bx,\bu)\geq -\gamma(b(\bx))
        \label{eqn:cbf_condition}
    \end{align}
    where $\dot{b}(\bx)$ denotes the time derivative of $b$ along the system trajectories \eqref{eqn:nonlinear_system}. We define the set $ K_{\text{CBF}}(\bx)=\{\bu\in\mathcal{U}:\;\dot{b}(\bx)\geq-\gamma(b(\bx))\}$ which consists of all control inputs $\bu$ which satisfy the condition \eqref{eqn:cbf_condition}. 
\end{definition}
Note that if the system \eqref{eqn:nonlinear_system} is control affine i.e. $F(\bx,\bu)=f(\bx)+g(\bx)\bu$ for smooth functions $f:\mathcal{X}\rightarrow\mathbb{R}^n$ and $g:\mathcal{X}\rightarrow\mathbb{R}^{n\times m}$, then the synthesis of safe control inputs that guarantee forward invariance for $\mathcal{S}_\bx$ can be transformed to the following QP:
\begin{align}
  \textbf{CBF-QP}\quad &\bu_S(\boldsymbol{x}):=\underset{\bu\in\mathbb{R}^m}{\text{argmin}}\quad\|\bu-k(\boldsymbol{x})\|^2\quad\text{s.t.} \quad L_fb(\bx)+L_gb(\bx)\bu\geq-\gamma(b(\bx)).
\label{eqn:quadratic_problem}
\end{align}
 where $L_f$ and $L_g$ denote the Lie derivatives of functions $f$ and $g$ respectively and $k(\bx)$ is a stabilizing feedback controller. Note that the function $\gamma$ within the CBF constraint \eqref{eqn:cbf_condition} indicates the degree to which safety is emphasized, thus allowing the agent to act either aggressively or cautiously. The CBF-QP controller addresses the problem of producing a control input sequence $\{\bu\}^T_{t=0}$, ensuring both convergence towards the goal and invariance of the set $\mathcal{S}_x$. However, if the QP \eqref{eqn:quadratic_problem} is infeasible at a particular time step $t$, it means that there is no valid solution considering the CBF constraints at that moment. In other words, the sequence of QPs becomes recursively infeasible. Consequently, the agent remains safe in its present position but stops its journey to the target. This can lead to a breakdown in the agent's navigation. Furthermore, the safe control synthesis is based on the assumption that the governing system is control affine in nature and also assumes that the stabilizing feedback controller $k(\bx)$ is known a-priori. In the following section, we briefly discuss Integral Control Barrier functions (ICBF's) that are able to encode both the state and input constraints in a single scalar function and also applicable to general nonlinear controlled systems. 

\subsection{Integral Control Barrier Functions (ICBF)\label{subsec:icbf}}
Consider the general integral feedback law which is defined by $\dot{\bu}=\phi(\bx, \bu)$ with $\bu(0):=\bu^0$
where $\phi: \mathcal{X} \times \mathcal{U} \rightarrow \mathbb{R}^m$ is a continuously differentiable function. Combining this feedback law with the \eqref{eqn:nonlinear_system} results in:
\begin{align}
[\dot{\bx} ,\;\dot{\bu}]^\mathrm{T}
=[F(\bx, \bu),\;\phi(\bx, \bu)]^\mathrm{T} ,\quad [\bx(0),\;\bu(0)]
=[\bx^0,\;\bu^0]^\mathrm{T}
\label{eqn:integrated_augmeted_system}
\end{align}
 We define $\bz:=\left[\bx^{\mathrm{T}}, \bu^{\mathrm{T}}\right]^{\mathrm{T}}$ as the augmented state of the integrated system and $\mathcal{S}:=\{(\bx,\bu)\in\mathcal{S}_{\bx}\times\mathcal{U}\}$ as a safety set that includes both state and input constraints, described by a scalar-valued function $h(\bz):=h(\bx,\bu)$ i.e. $\mathcal{S}=\{(\bx,\bu):\;h(\bx,\bu)\geq 0\}$. For every trajectory of the integrated system, the following must hold $\dot{h}(\bz)+\gamma(h(\bz)) \geq 0 $ to guarantee forward invariance.
However, typical CBF techniques are not effective here due to the integral control law. To address this, an auxiliary input $\bv$ is introduced in \cite{ames2020integral_cbf}. Consequently, the integral law modifies to $\dot{\bu}=\phi(\bx,\bu)+\bv$ where $\bv$ is designed to guarantee forward invariance for set $\mathcal{S}$.

\begin{definition}
\normalfont An Integral Control Barrier Function (ICBF) is a function $h:\mathcal{X}\times\mathcal{U}\rightarrow\mathbb{R}$ that determines a safe set $\mathcal{S}$. For it to be an ICBF, for all $(\bx,\bu)$ such that $p(\bx, \bu)=0$ then $q(\bx, \bu) \leq 0$ where
\begin{align}
&p(\bx, \bu)  :=\left(\frac{\partial h(\bx,\bu)}{\partial \bu}\right)^{\mathrm{T}},\; q(\bx, \bu):  =-\left(\frac{\partial h(\bx,\bu)}{\partial \bx} F(\bx, \bu)+\frac{\partial h(\bx,\bu)}{\partial \bu} \phi(\bx, \bu) +\gamma(h(\bx, \bu))\right)\nonumber
\end{align}
\end{definition}
Note that the notion of ICBF can also be defined for nonlinear systems with higher relative degree i.e. for systems where $p(\bx,\bu)=0$ does not imply $q(\bx,\bu)\leq 0$ \cite{zinage2023disturbance}. For the ICBF $h$, safe control inputs are synthesized from the following theorem.
\begin{theorem}
\normalfont [\cite{ames2020integral_cbf}] If there exists an integral feedback controller $\phi(\bx,\bu)$ and a safety set $\mathcal{S}$ defined by an ICBF $h(\bx,\bu)$, then tweaking the integral controller to $\dot{\bu}=\phi(\bx, \bu)+\bv^\star(\bx, \bu)$
where $\bv^\star(\bx, \bu)$ is obtained by solving the following QP:
\begin{align}
\bv^\star(\bx, \bu) = \underset{\bv \in \mathbb{R}^m}{\operatorname{argmin}}\|\bv\|^2\quad
\text {subject to }\;\; p(\bx, \bu)^{\mathrm{T}} \bv \geq q(\bx, \bu)
\label{eqn:qp_icbf}
\end{align}
ensures safety i.e. forward invariance for set $\mathcal{S}$. 
\begin{remark}
    \normalfont Note that to leverage ICBF to compute safe control inputs for non-control affine (or control affine) nonlinear systems, one needs to know the integral feedback control law $\phi(\bx,\bu)$ a-priori, which is in general challenging, and also the nonlinear dynamics \eqref{eqn:nonlinear_system}, which might be changing with time due to external disturbances. Furthermore, there must be a mechanism to encode both the state/safety and input constraints on a single scalar-valued function $h(\bx,\bu)$ which in general is non-trivial. Our proposed approach in the following subsections addresses these challenges.
\end{remark}
\end{theorem}
\begin{figure}[ht]
 \centering
\includegraphics[width=0.6\textwidth]{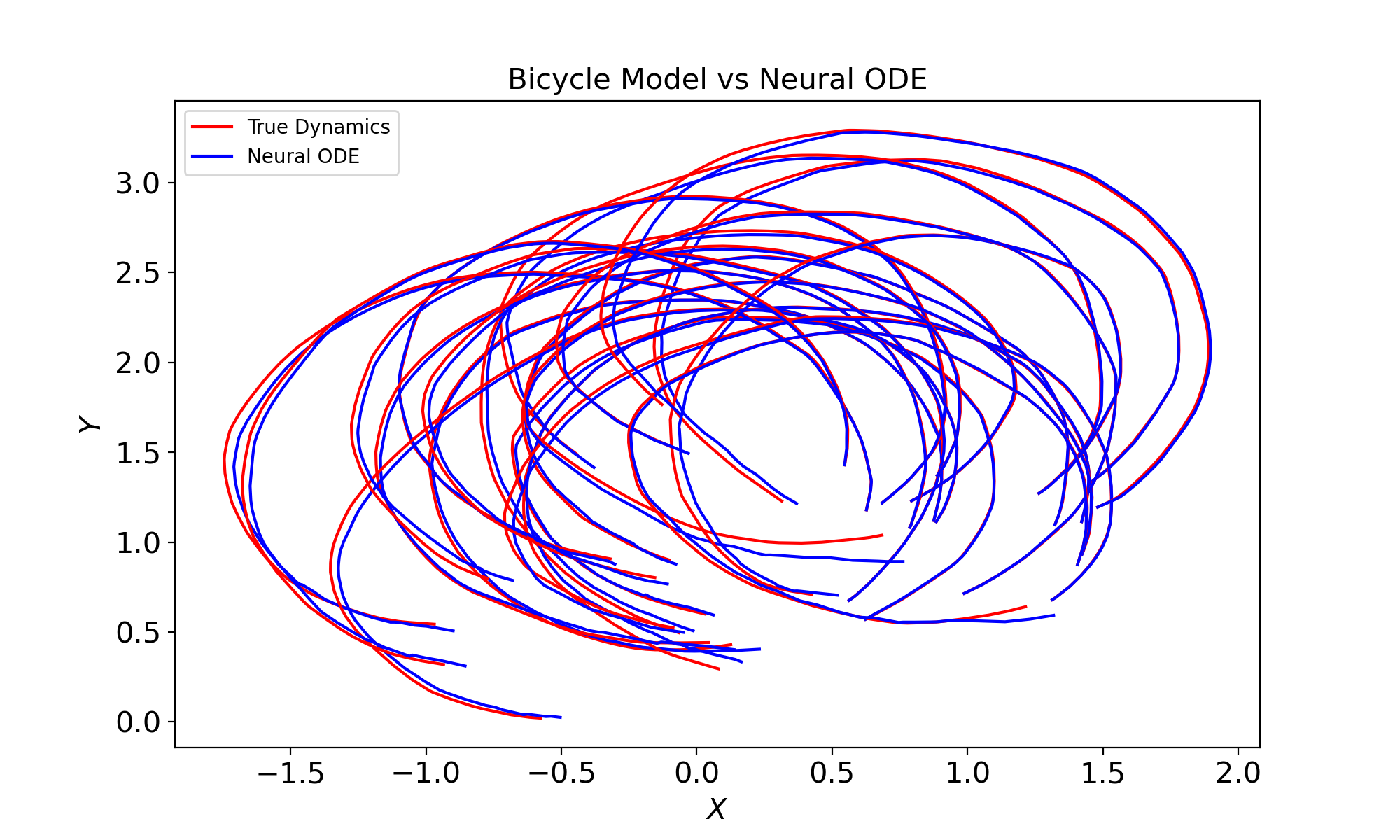}
 \caption{\small The figure shows the long-term prediction from the learned neural ODE model $F_\theta$ (blue) and the actual vehicle dynamics $F$ (red) over the time span of $10s$. The initial states for the system are sampled from a ball of radius 1 and the random control sequence are set of control inputs uniformly sampled from $[-1,1]^2$. 
 }
\label{fig:neural_ode}
\vspace{-10pt}
\end{figure}
\paragraph{Problem Statement:} Under the assumption that a valid ICBF, the governing system dynamics, and the integral control law are not known a-priori, design a sequence of control inputs that guarantee safety for the nonlinear system \eqref{eqn:nonlinear_system} and satisfaction of input constraints (i.e. $\bu\in\mathcal{U}$)
\section{Proposed approach\label{sec:main_result}}

In this section, we propose a deep learning framework to simultaneously learn $(i)$ the system dynamics $F(\bx,\bu)$, $(ii)$ the scalar-valued ICBF $h(\bx,\bu)$ (that allows encoding both the state and input constraints), and $(iii)$ the imitation-based stabilizing feedback controller $k(\bx)$. Finally, these three learned models are integrated to design safe neural network based feedback controllers $\bu_{\text{safe}}^\star(\bx)$ \eqref{eqn:u_safe_neural} for general nonlinear systems in the presence of input constraints.
\paragraph{Nonlinear system identification:}
\label{subsec:nonlinear_system_identification}
We leverage Neural ODEs \cite{chen2018neural_ode} that offer a promising approach to learning the controlled dynamics $ F $ in \eqref{eqn:nonlinear_system}. The essence of Neural ODEs is to design a neural network that, given a state $\bz$, predicts its derivative $\dot{\bz}$.
During training, a loss function, such as the Mean Squared Error (MSE) between the network's predicted derivative and the actual derivative, is minimized using standard optimization techniques. Once trained, this Neural ODE, in conjunction with standard ODE solvers, can simulate the system's \eqref{eqn:nonlinear_system} future trajectories from any given initial state $\bx^0$. We assume the availability of a finite dataset $\mathcal{D}_F$, comprised of system trajectories that serve as time-series data of states and corresponding control inputs. Formally, the dataset is defined as $\mathcal{D}_F = \{ \Gamma_1, \dots, \Gamma_m\}$, where each trajectory $\Gamma_i$ for all $i=[1,m]_d$ (where $[a,b]_d$ for $b>a$ denotes the set of integers $\{a,a+1,\dots,b\}$) is a sequence $\{(\bx_0, \bu_0), \dots,(\bx_N, \bu_N)\}$. In this sequence, $\bx_i = \bx(t_i)$ represents the state at time $t_i$, $\bu_i = \bu(t_i)$ is the control input at time $t_i$, and $t_{i-1}<t_i$ where $i\in[1,N]_d$.
 The Neural ODE is trained by minimizing the discrepancy between the predicted states and the observed states from the dataset $\mathcal{D}_F$. The loss function $\mathcal{L}_F(\theta)$ can be formulated as:
\begin{align}
    \mathcal{L}_F(\theta) =  \sum_{\Gamma \in \mathcal{D}_F} \sum_{i=1}^{N} \left\| \bx_i - \Phi_{F_{\theta}}(\bx_{i-1}, \bu_{i-1}, \Delta t) \right\|^2,
    \label{eqn:loss_nonlinear_system}
\end{align}
where $\theta$ denotes the parameters of the neural network $F_\theta$ and $\Phi_{F_{\theta}}$ (obtained by integrating $F_\theta$ over $\Delta t$ interval) denotes the state transition map induced by the Neural ODE over a time interval $\Delta t$, $F$ is the actual nonlinear dynamics, and $\|\cdot\|$ represents the Euclidean norm.
Upon successful training, the Neural ODE model can be utilized for long-term predictive tasks. Given a current state $\bx_k$ and a sequence of control inputs $\bu_k, \dots, \bu_{k+r}$, the model can forecast future states $\bx_{k+1}, \dots, \bx_{k+r+1}$ by integrating the learned dynamics i.e., $    \bx_{k+i+1} = \Phi_{F_{\theta^\star}}(\bx_{k+i}, \bu_{k+i}, \Delta t)$ for $i = [0,r]_d$.
This predictive capability enables the application of the learned model to control and planning tasks, where accurate long-term future state estimation is crucial.


\paragraph{Imitation learning for Nonlinear MPC:}
\label{subsec:imitation_nonlinear_mpc}

In this section, we approximate the Nonlinear Model Predictive Control (NMPC) policies by a neural network via imitation learning. For general nonlinear controlled systems (especially non-control affine systems), solving the optimization problem in real-time can be computationally expensive as they involve solving a sequence of nonlinear programs (NLPs). Imitation learning (IL) offers a way to learn an approximation of the optimal control strategy from expert demonstrations (optimal policies from NMPC in this case), thereby speeding up the control computation, especially for real-time applications such as off-road autonomy. 
The objective of NMPC is to minimize the cost function over a finite horizon $T$:
    \begin{align}
\underset{\substack{\bx_0, \bu_0, \ldots, \bx_{N_T}}}{\min}   \sum_{k=0}^{N_T-1} \tilde{L}\left(\bx_k, \bu_k\right) \;\;
\text { s.t. }\;\;  \bx_0=\bx^0, \;\; \bx_{k+1}=F_d\left(\bx_k, \bu_k\right),\;\; k=[0, N_T-1]_d
\label{eqn:standard_mpc}
\end{align}
where $N_T$ is the time horizon, ${\bx}^0$ is the initial condition and $F_d$ is obtained after applying fourth order Runga-Kutta discretization scheme to \eqref{eqn:nonlinear_system}. 
Note that in this formulation, we have not considered the collision and the input/state constraints. As will be discussed in later sections, these constraints would be encoded via learned Neural Integral Control Barrier Functions (Neural ICBFs). To leverage IL, a dataset $\mathcal{D}_c$ of optimal state-input trajectories i.e. $\mathcal{D}_c:=\{(\bx^\star_i, \bu_i^\star)\}_{i=1}^{N_c}$ is collected using the NMPC controller (solving \eqref{eqn:standard_mpc}), where $\bu_i^\star$ represents the optimal control input for current state $\bx^\star_i$. A neural network approximator $\pi:\mathcal{X}\rightarrow\mathbb{R}^m$ is then trained to map from states to optimal inputs i.e. $\hat{\bu}(\bx) = \pi_\theta(\bx)$
where $\theta$ are the parameters of the neural network model and $\hat{\bu}(\bx)$ is the optimal solution for \eqref{eqn:standard_mpc}.
The training objective is to minimize the difference between the predicted controls and the expert controls (from the NMPC controller).  In other words, we minimize the  loss function $\mathcal{L}_c(\theta)=\mathbb{E}_{\bx\sim \mathcal{D}_c}[\ell(\bx,\pi_\theta(\bx))]$
where $\mathcal{D}_c$ is the given state distribution, $\ell$ is the pointwise loss function for $\pi_\theta(\bx,\theta)$ for a given $\bx$. The final goal of the IL based framework is to learn the optimal set of parameters $\theta^\star_c$ that minimizes the expected loss i.e. $\theta^\star_c=\underset{\theta}{\text{argmin}}\;\mathcal{L}_c(\theta)$.
Once trained, the function approximator $\pi_\theta$ can be used to produce control inputs in real time, eliminating the need to solve the NMPC optimization problem at each timestep.
\begin{figure}[ht]
 \centering
\includegraphics[width=0.9\textwidth]{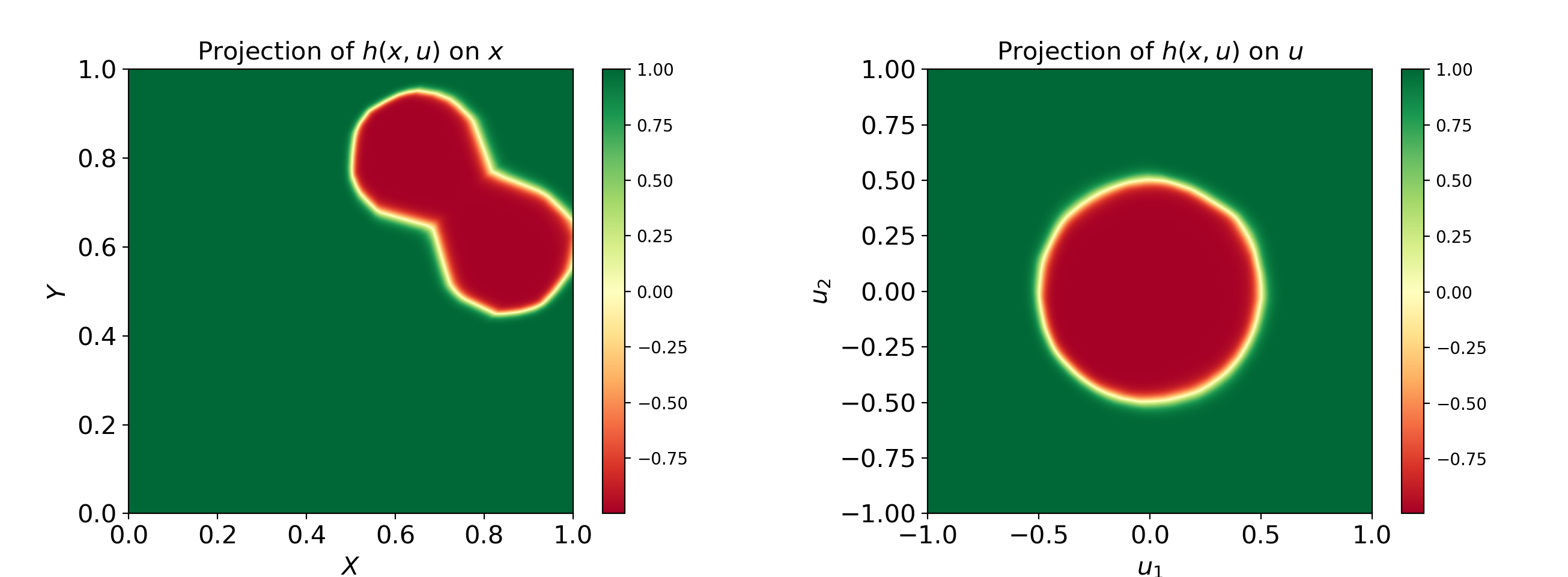}
 \caption{\small Projection of learned $h(\bx,\bu)$ on the $\bx$ and $\bu$ 2D planes for the vehicle dynamics based on bicycle model. As seen from the figure, our approach is able to encode both the state/safety and input constraints onto a single scalar-valued function of $\bx$ and $\bu$. The input constraint for the vehicle is that $\bu$ must belong to a ball of radius $0.5$ centered around the origin i.e. $\mathcal{U}=\{\bu:\; \|\bu\|^2\leq 0.25\}$ and the state constraints corresponding to two circular obstacles. As seen from the figure, the value of $h(\bx,\bu)$ on the boundary of regions is 0 ($\text{white}$), inside the unsafe region is negative ($\text{red}$), and in the safe region is positive ($\text{green}$).}
\label{fig:icbf}
\vspace{-10pt}
\end{figure}

\paragraph{Neural Integral Control Barrier Functions (Neural ICBF's):}
\label{subsec:neural_icbfs}
To learn a ICBF for \eqref{eqn:nonlinear_system}, we leverage a neural network based representation, denoted by $h_\theta(\bx,\bu):\mathcal{X}\times\mathcal{U}\rightarrow\mathbb{R}$:
where $\theta$ are the parameters of the neural network $h_\theta$. 
The objective is to ensure that the Lie derivative of the learned function along the learned nonlinear system is positive. Furthermore, $h_\theta(\bx,\bu)$ must be positive and non-positive if $(\bx,\bu)$ are safe and unsafe, respectively.  To achieve this, we define the following empirical loss function $\mathcal{L}_h(\theta)$:
\begin{align}
\small
\mathcal{L}_h(\theta) = \sum_{i=1}^{N_h} (-p_\theta(\bx^\text{S}_i,\bu^\text{S}_i)^\mathrm{T}\bv^\star(\bx, \bu)+q_\theta(\bx^\text{S}_i,\bu^\text{S}_i)-\epsilon)_++\sum_{i=1}^{N_s}(-h_\theta(\bx^\text{S}_i,\bu^\text{S}_i))_+
+\sum_{i=1}^{N_h-N_s}(h_\theta(\bx^\text{US}_i,\bu^\text{US}_i))_+
\label{eqn:loss_h}
\end{align}
where $x_+=\max\{0,x\}$, $N_h$ is the total number of state-input pairs, $N_s(<N_h)$ is the number of safe state-input pairs, $\bx^\text{S}_i\in\mathcal{S}_\bx$, $\bx^\text{US}_i\in\mathcal{X}\setminus\mathcal{S}_\bx$, $\bx^\text{S}_i\in\mathcal{U}$, $\bu^\text{US}_i\notin \mathcal{U}$, $\epsilon$ is an upper bound for the absolute error between learned and actual models (Theorem \ref{thm:app_error}), $\alpha(.)$ is a class-$\mathcal{K}_\infty$ function and $p_\theta(\bx, \pi_\theta(\bx) ) $ and $q_\theta(\bx, \pi_\theta(\bx) ) $ are given by
\begin{align}
\small
&p_\theta(\bx, \pi_\theta(\bx) ) :=\left(\frac{\partial h_\theta(\bx,\pi_\theta(\bx))}{\partial \bu}\right)^{\mathrm{T}},\nonumber\\ & q_\theta(\bx, \pi_\theta(\bx)):  =-\left(\frac{\partial h_\theta(\bx,\pi_\theta(\bx))}{\partial \bx} F_{\theta}(\bx, \pi_\theta(\bx))+\frac{\partial h_\theta(\bx,\pi_\theta(\bx))}{\partial \bu} \phi(\bx, \pi_\theta(\bx)) +\gamma(h_\theta(\bx, \pi_\theta(\bx))\right)\nonumber
\end{align}
Since the NN is trained on a finite set of data points, it is important to recognize that it may not meet the ICBF-based conditions throughout $\mathcal{X}\times\mathcal{U}$, even if the $\mathcal{L}_h(\theta)$ approaches zero. Towards that goal, SMT solvers \cite{zinage2023neural_koopman} (computationally expensive) and recently proposed faster verification methods \cite{wang2023simultaneous_verification_cbf_faster} can be leveraged.

\begin{remark}
 \normalfont   Note that it is possible to extend this framework to learn ICBFs for nonlinear systems with time-varying additive disturbances i.e. systems of the form $\dot{\bx}=F(\bx,\bu)+\ell(\bx)\bd(t)$. In that case, one can utilize, for instance, Disturbance Observer based ICBF (DO-ICBF) presented in \cite{zinage2023disturbance} and consequently an upper bound for the absolute error between the learned/trained safe input and the actual input can be derived using Theorem 3 in \cite{zinage2023disturbance} and Theorem \ref{thm:app_error}. However, for brevity and due to page limits we do not consider the case of nonlinear systems with additive disturbances. 
\end{remark}


\paragraph{Implementation:}
\label{sec:implementation}
The algorithm for safe control synthesis of the general unknown nonlinear system with input constraints is given in Algorithm \ref{alg:algorithm}. The main steps of Algorithm \ref{alg:algorithm} are as follows. First, the NN models for ICBF $h_{\theta}(\bx,\bu)$, system dynamics $F_{\theta}(\bx)$ and imitation learner $\pi_\theta$ that mimics NMPC is learned using the datasets $\mathcal{D}_h$, $\mathcal{D}_F$ and $\mathcal{D}_\pi$ respectively. This is given in Lines $25-26$ of Algorithm \ref{alg:algorithm}.
The iterative cycle for learning these NN models persists until the accuracy in the test data is satisfied to a predefined degree of accuracy. The function $\texttt{MAIN}$ in Line $24$ of Algorithm \ref{alg:algorithm} uses the learned NN models $\pi_\theta(\bx)$, $F_\theta(\bx,\bu)$ and $h_\theta(\bx,\bu)$ to synthesize a safe feedback controller $\bu^\star_{\text{safe}}(\bx)$ (Line $29$ of Algorithm \ref{alg:algorithm}) that guarantees safety (forward invariance) and also satisfies the input constraints (given by $\bu\in\mathcal{U}$). The learned safe controller $\bu_{\text{safe}}^\star(\bx)$ is analytically given by
\begin{align}
\bu_{\text{safe}}^\star(\bx)=\pi_\theta(\bx)+\begin{cases}\int_{t}^{t+\Delta t}\frac{q_\theta(\bx, \pi_\theta(\bx))}{\|p_\theta(\bx,\pi_\theta(\bx))\|^2} p_\theta(\bx,\pi_\theta(\bx)) \mathrm{dt} & \text { if } q_\theta(\bx, \pi_\theta(\bx))  >0 \\ 0 & \text { if }q_\theta(\bx, \pi_\theta(\bx)) \leq 0\end{cases}
\label{eqn:u_safe_neural}
\end{align}
where $\Delta t>0$ is the sampling time period.
It can be easily shown that the analytical solution to QP \eqref{eqn:qp_icbf} (for fixed $\bx$ and $\bu$ and for $p(\bx,\bu)\neq 0$) is given by $\bv^\star(\bx, \bu)=\frac{q(\bx, \bu)}{\|p(\bx, \bu)\|^2} p(\bx, \bu)$ if $q(\bx, \bu)>0$ otherwise $0$.
\begin{theorem}
    \normalfont Let $(\bx,\bu)\in\mathcal{X}\times\mathcal{U}$ and $(\by,\boldsymbol{w})_\mathcal{D}$ is the state-input pair for the dataset $\mathcal{D}$. If $F$, $F_\theta$, $h$, $h_\theta$, $\phi$, $\pi$ are continuous with Lipschitz constants $L_F,\;L_{F_\theta},$ $L_h,\;L_{h_\theta},\;L_{\phi},\;L_{\phi}$ respectively, then
    \begin{align}
        \|\bu^\star_\theta-\bu^\star\|\leq E_2(\phi,\pi_\theta)+\Delta t(3E_3(h,h_\theta)+E_1(F,F_\theta)+E_2(\phi,\pi_\theta)+E_1(\gamma(h),\gamma(h_\theta)))
        \label{eqn:thm_eqn_ineq}
    \end{align}
    where $E_1(M,N):=L_M\delta_1(M)+L_N\delta_1(N)+\mu(M,N)$, $E_2(M,N):=L_M\delta_2(M)+L_N\delta_2(N)$, $E_3(M,N):=L_M+L_N$, $\delta_1(N):=\underset{(\bx,\bu)\in\mathcal{D}_N}{\max}\;\|(\bx,\bu)-(\by,\bw)\|$,  $\delta_2(N):=\underset{\bx\in\mathcal{D}_N}{\max}\;\|\bx-\by\|$, $\mu(M,N):=\underset{(\bx,\bu)\in\mathcal{D}_N}{\max}\;\|M-N\|$ and $N$ is the model that approximates the actual model $M$.
    \label{thm:app_error}
\end{theorem}
\begin{proof}
    \normalfont Let $(\bx,\by)\in\mathcal{X}\times\mathcal{U}$ and $(\by,\boldsymbol{w})_\mathcal{D}$ be a sample belonging to the dataset $\mathcal{D}$. The absolute error in $\phi$ can be upper bounded as follows:
    \begin{align}
        \|\phi(\bx)-\pi_\theta(\bx)\|&\leq \|\phi(\bx)-\phi(\by)\|+\|\phi(\by)-\pi_\theta(\by)\|+\|\pi_\theta(\by)-\pi_\theta(\bx)\|\nonumber\\
&\leq L_\phi\delta_2(\phi)+\mu(\phi,\pi_\theta)+L_{\pi_\theta}\delta_2(\pi_\theta):=E_2(\phi,\pi_\theta)\label{eqn:error_phi}
    \end{align}
    Similarly, it can be shown that the following inequalities hold true
    \begin{align}
    \small
       \|F(\bx,\bu)-F_\theta(\bx,\bu)\| \leq E_1(F,F_\theta),\; \left\|\frac{\partial \Delta h(\bx,\pi_\theta(\bx))}{\partial \bu}\right\|\leq E_3(h,h_\theta),\; \left\|\frac{\partial \Delta h(\bx,\pi_\theta(\bx))}{\partial \bx}\right\|\leq E_3(h,h_\theta)\nonumber
    \end{align}
    where $\Delta h(\bx,\pi_\theta(\bx))=h(\bx,k(\bx))-h_\theta(\bx,\pi_\theta(\bx))$. Consequently, we have
    \begin{align}
    \small
        p_\theta(\bx, \pi_\theta(\bx))\leq E_3(h,h_\theta),\quad q_\theta(\bx, \pi_\theta(\bx))\leq 2E_3(h,h_\theta)+E_1(F,F_\theta)+E_2(\phi,\pi_\theta)+E_1(\gamma(h),\gamma(h_\theta))\label{eqn:error_p_and_q}
    \end{align}
    Therefore using \eqref{eqn:u_safe_neural}, \eqref{eqn:error_phi}, and \eqref{eqn:error_p_and_q},  \eqref{eqn:thm_eqn_ineq} holds.
\end{proof}

\section{Results}
\label{sec:results}
In this section, we benchmark nonlinear control examples. First, we consider the non-control affine vehicle dynamics based on Ackerman steering and subsequently highly nonlinear pendulum on a quadrotor system. For both these examples, we use the $\texttt{do-mpc}$ package to solve the NLP's for NMPC for generating expert trajectories $\mathcal{D}_u$. All the code is implemented in Pytorch. We compare our approach with ORCA-MAPF \cite{orcamapf} that integrates decentralized collision avoidance facilitated by ORCA with centralized Multi-Agent Path Finding (MAPF) and reinforcement learning-based method CADRL \cite{cadrl}.
We compare our approach with respect to two performance metrics. First, is the average computational time required to generate a control sequence for a nonlinear system to reach from uniformly sampled $100$ initial states to the desired goal states in the presence of $M$ obstacles. Second, is the average cost of the solution returned by the algorithm. The cost function is given by
\begin{align}
c(\bx^0)=&\underset{\bu_i\in\mathcal{U}}{\min}\;\sum_{i=1}^N\bx_i^\mathrm{T}Q\bx_i+\bu^\mathrm{T}_iR\bu_i,\;\; \text{s.t.}\;\;{\bx}_{k+1}=F_d(\bx_k,\bu_k),\quad \bx_0=\bx^0,\quad k\in[0,N-1]_d
\label{eqn:cost_function}
\end{align}
where $\bx^0$ is the initial condition, $Q\succ 0$, $R\succ 0$ and $T$ is time horizon. 
The average cost is defined by $\sum_{i=1}^{N_{\text{avg}}}c(\bx^0_i)/N_{\text{avg}}$ where $\bx^0_i$ are randomly sampled initial conditions. We choose $N_{\text{avg}}=100$.
\begin{minipage}{.48\textwidth}
\paragraph{Simulation setup:} The dynamics of the vehicle are non-control affine as in \cite{zinage2023disturbance}.
The objective for the vehicle is to avoid a set of $M$ circular obstacles in 2D plane and reach the goal 
 state while minimizing the cost function \eqref{eqn:cost_function}. For the quadrotor-pendulum system, we consider the dynamics presented in \cite{beard2008quadrotor}. The goal of the quadrotor-pendulum system is to also avoid $M$ spherical obstacles placed randomly in the 3D workspace to reach the desired goal position.
\vspace{0.3cm}
\paragraph{Trajectory data generation:}
We randomly sample $200$ initial states $\bx^0$ in the regions $[0,1]^n$ for the vehicle dynamics ($n=2$) and the quadrotor-pendulum system ($n=12$). Nonlinear MPC is then employed to compute the sequence of optimal inputs from sampled initial states $\bx^0$ to the origin.
These inputs are then used by the imitation learner based neural $\pi_\theta$ network to mimic the NMPC. For learning $h_\theta$ we sample $10,000$ state-input pairs for both systems. For learning $F_\theta$, we choose $\Delta t=0.1s$ and $T=10s$ and simulate $200$ trajectories for time horizon $T$ with random inputs in $[0,1]^m$ where $m=2$ and $m=4$ for vehicle and quadrotor-pendulum system respectively.
\vspace{0.3cm}

\paragraph{Methodological Configuration.}
The architecture of neural networks $\pi_\theta(\bx),\; F_\theta(\bx,\bu)$ and $h_\theta(\bx,\bu)$, comprises a multilayer perceptron (MLP) with $2$, $4$, and $4$ intermediary layers respectively for vehicle system and $2$, $4$, and $4$ intermediary layers respectively for the quadrotor-pendulum system, each consisting of

\end{minipage}
\hfill
\begin{minipage}{.48\textwidth}
\begin{algorithm}[H]
\small
\caption{Proposed approach}
\label{alg:algorithm}
\textbf{Input}: $\mathcal{D}_F$, $\mathcal{D}_h$, $\mathcal{D}_c$ (Section \ref{sec:main_result})\\
\textbf{Parameter}: $\Delta t$, $N_d$,  
$N_T$, $N_h$ (Section \ref{sec:main_result})\\
\textbf{Output}: $F_\theta$, $h_\theta$, $\pi_\theta$, $\bu^\star_{\text{safe}}(\bx)$\\
\begin{algorithmic}[1] 
\STATE $\textbf{function}\;\;\texttt{SYSTEM}(\mathcal{D}_F)$
  \quad\quad\FOR{each trajectory $\Gamma$ in $\mathcal{D}_F$}
        \FOR{time step $i$ from $1$ to $T$}
            \STATE $\hat{\bx}_i \leftarrow \Phi_{F_{\theta}}(\bx_{i-1}, \bu_{i-1}, \Delta t)$
            \STATE $\mathcal{L}_F(\theta) \leftarrow \mathcal{L}_F(\theta)+\left\| \bx_i - \hat{\bx}_i \right\|^2$
        \ENDFOR
    \ENDFOR 
\STATE $\quad\textbf{Repeat:}$
\STATE $\quad\quad\quad\theta\leftarrow\theta-\alpha\nabla\mathcal{L}_F(\theta)$
\STATE $\quad\textbf{until}\; \text{convergence}$
\STATE $\textbf{end function}$
\STATE $\textbf{function}\;\;\texttt{ICBF}(\mathcal{D}_h,F_{\theta^\star})$
\STATE $\quad\quad\quad\mathcal{L}_h(\theta)\leftarrow \text{Eqn. } \eqref{eqn:loss_h}$
\STATE $\quad\textbf{Repeat:}$
\STATE $\quad\quad\quad\theta\leftarrow\theta-\alpha\nabla\mathcal{L}_h(\theta)$
\STATE $\quad\textbf{until}\; \text{convergence}$
\STATE $\textbf{end function}$
\STATE $\textbf{function}\;\;\texttt{Controller}(\mathcal{D}_u)$
\STATE $\quad\quad\quad \mathcal{L}_c\leftarrow \mathbb{E}_{\bx\sim \mathcal{D}}[\ell(\bx,\phi(\bx;\theta))]$
\STATE $\quad\textbf{Repeat:}$
\STATE $\quad\quad\quad\theta\leftarrow\theta-\alpha\nabla\mathcal{L}_h(\theta)$
\STATE $\quad\textbf{until}\; \text{convergence}$
\STATE $\textbf{end function}$
\STATE $\textbf{function}\;\;\texttt{MAIN}$
\STATE $F_\theta\leftarrow\texttt{SYSTEM}(\mathcal{D}_F),\;\;h_\theta\leftarrow\texttt{ICBF}(\mathcal{D}_h)$
\STATE $\pi_{\theta}\leftarrow\texttt{Controller}(\mathcal{D}_u)$
\STATE $\bx\leftarrow\bx^0,\quad\bu\leftarrow\pi_{\theta}(\bx^0),\quad t\leftarrow 0$
\WHILE{$\|\bx-\bx_{\text{goal}}\|\leq\epsilon_1 (\text{threshold})$}
\STATE $\bu^\star_{\text{safe}}(\bx)\leftarrow\text{Eqn. }\eqref{eqn:u_safe_neural}$
\STATE $\bx(t+\Delta t)\leftarrow \text{Integrate }F_{\theta}(\bx(t),\bu^\star_{\text{safe}}(\bx))$
\STATE $t\leftarrow t+\Delta t$
\ENDWHILE
\STATE $\textbf{end function}$
\end{algorithmic}
\end{algorithm}
\end{minipage}
128 nodes, employing $\texttt{ReLu}$ as the activation function for $\pi_\theta(\bx)$ and $h_\theta(\bx,\bu)$ and $\texttt{tanh}$ activation function for $F_\theta(\bx,\bu)$.
For learning $F$, the neural ODE-based loss function \eqref{eqn:nonlinear_system} incorporates a future prediction span $T$ of $10s$ with sampling time $\Delta t=0.1s$ and $1000$ trajectories. For learning $h_\theta(\bx,\bu)$, a set of $10,000$ points are randomly distributed across the domain of the state space $\mathcal{X}$ and classified as safe or unsafe for training. Furthermore, for learning IL based controller $\pi$ we use $1000$ expert trajectories generated from NMPC. Lastly, we use an $\texttt{Adam}$ optimizer for training these neural networks.
\begin{table}[H]
\small
\centering
\begin{minipage}{.5\linewidth}
\centering
\begin{tabular}{|c|c|c|}
\hline
Method & Comp. time (in s) & Cost \\
\hline
NMPC & 5.562 & 4.583 \\
\hline
ORCA-MAPF & 1.373 & 6.839 \\
\hline
CADRL & 0.837 & 5.683 \\
\hline
Our approach & 0.736 & 4.631 \\
\hline
\end{tabular}
\caption{Vehicle system
}
\label{tab:vehicle}
\end{minipage}%
\begin{minipage}{.5\linewidth}
\centering
\begin{tabular}{|c|c|c|}
\hline
Method & Comp. time (in s) & Cost \\
\hline
NMPC & 10.242 & 15.232 \\
\hline
ORCA-MAPF & 2.328  & 21.281 \\
\hline
CADRL & 2.181  & 19.283 \\
\hline
Our approach & 1.550  & 15.672\\
\hline
\end{tabular}
\caption{Quadrotor-pendulum system}
\label{tab:quadrotor}
\end{minipage}
\end{table}
\paragraph{Discussion:} 


The comparative analysis presented in Table \ref{tab:vehicle} highlights the efficacy of the proposed approach over various classical and learning-based control methods. Notably, NMPC requires the most computational time yet achieves optimal cost as it is based on an optimization-based framework. This is reasonable computational time as NMPC solves a series of non-convex NLP (non-convex because of the collision avoidance constraints) problems that are computationally expensive. In contrast, ORCA-MAPF and CADRL demonstrate improved computational efficiency, but at the expense of slightly higher optimal costs compared to our approach. This would be mainly because CADRL uses a heuristic sparse reward function depending on whether the agent has reached the goal or not and ORCA-MAPF is based on carefully-designed strategies based on velocity cones for collision avoidance. Our method, however, excels by requiring the least average computational time while also resulting in an average cost that is close to the average optimal cost (from NMPC), indicative of its potential for real-time applications where efficiency and optimal cost of safe trajectories are paramount. This is mainly because our proposed approach minimally deviates from the optimal trajectory (generated in the absence of obstacles) which is obtained by solving the QP \eqref{eqn:qp_icbf} resulting in a safe trajectory in the presence of obstacles. Consequently, our approach suggests a significant improvement in computational performance from classical NMPC and learning based methods, with a modest trade-off in terms of average cost.

\section{Conclusion\label{sec:conclusion}}
In this paper, we propose a real-time deep learning framework for generating safe control inputs for nonlinear controlled systems with unknown dynamics in the presence of input constraints. This framework simultaneously learns the system dynamics, the integral control law for precise trajectory tracking, and the Integral Control Barrier Functions (ICBFs) which correspond to safety certificates that can simultaneously encode both the state and input constraints into a scalar-valued function. The numerical simulations conducted demonstrate the effectiveness and efficiency of our proposed method, with respect to the computational time and the cost of the trajectories generated. Future work could focus on extending this framework to vision-based tasks, scalability for multi-agent motion planning problems and implementation of the proposed approach onto real vehicles for off-road navigation tasks. 

\section{Acknowledgements}
The authors would like to thank Cyrus Neary of the University of Texas at Austin for his assistance and provision of the Tikz code.
\bibliography{main}

\end{document}